\documentclass[pra,twocolumn,a4paper,showpacs,floatfix]{revtex4}
\usepackage{latexsym,amssymb,amsmath,graphicx,amsthm}
\usepackage{amsthm}
\begin{document}

\title{Nonlinear coherent loss for generating non-classical states}

\author{ A. Mikhalychev, D. Mogilevtsev and S. Kilin}

\affiliation{Institute of Physics, Belarus National Academy of
Sciences, Nezavisimosty Ave. 68, Minsk 220072 Belarus }

\begin{abstract}
Here we discuss generation of non-classical  states of bosonic
mode with the help of artificially designed loss, namely the
nonlinear coherent loss. We show how to generate superpositions of
Fock states, and how it is possible to "comb" the initial states
leaving only states with certain properties in the resulting
superposition (for example, a generation of a superposition of
Fock states with odd number of particles). We discuss purity of
generated states and estimate maximal achievable generation
fidelity.
\end{abstract}

 \maketitle

\section{Introduction}

Artificially designed nonlinear loss is a rather novel and highly
promising way to generate deterministically non-classical states.
Nonlinear loss possesses a number of advantages over more usual
nonlinear unitary dynamics. First of all, it can turn a mixed
initial state into a final pure one, which is, generally,
impossible with unitary transformations. Then, the artificial loss
is, actually, turning a greatest enemy of the non-classical state
generators into their biggest ally. In schemes of generation of
non-classical states via unitary dynamics linear losses are prone
to destroy quantum features rather quickly (here one might recall
famous "Schroedinger-cat states" turning into mixtures with the
rate proportional to the average particle number of the state
\cite{cats}). However, nonlinear losses can force the system into
a certain non-vacuum state ("pointer state"
\cite{{zurek},{zurek1}}), and make it remain there despite
possible perturbations (which was epitomized in a concept of
"protecting" quantum states
\cite{{davidovich1},{moya-cessa-2000},{davidovich-2001}}). Then,
the nonlinear loss generation scheme can be extremely robust, and
is able to turn a wide variety of initial states into a predefined
output state \cite{krauss}. For example, one can produce
"Schroedinger-cat states" \cite{garraway}, and Fock states
\cite{poyatos} this way. The nonlinear loss can be a flexible and
efficient tool of quantum computation \cite{cirac}.

Recently, significant progress was achieved in development of
practical realizable schemes of artificial nonlinear loss. It was
shown that it is possible to produce a wide variety of nonlinear
losses for vibrational states of ions in magnetic traps
\cite{{zoller},{wineland}}, or atoms in optical lattices
\cite{zwerger}. Using a combination of optical nonlinearities, one
can produce nonlinear losses of electromagnetic field modes
\cite{{yamamoto},{bose},{ezaki},{hong}}. Implementing correlated
linear loss, one can also produce an effective nonlinear loss in
Bose-Einstein condensates \cite{valera}. Correlated loss gives an
opportunity to design nonlinear losses for the generation of
single-photon states and strongly sub-Poissonian states in
multi-core optical fibers \cite{mogilevtsev opt lett}.

This work is devoted to analysis of possible states which can be
produced with the help of designed nonlinear coherent loss quite
commonly encountered in the schemes mentioned above. This kind of
loss is able to produce a wide range of non-classical states. We
discuss conditions and prerequisites for generation of such
states, provide an estimation of fidelity and purity of generated
states. In particular, we show how to generate an arbitrary Fock
state with fidelity arbitrarily close to unity using nonlinear
coherent loss and an input coherent state. Also, we analyze the
generation of finite or countable superpositions of Fock states.
It appears that one cannot reach a unit fidelity for generation
of superposition of two Fock states. The fidelity of such
superposition generation gets worse when the states of the pair
are strongly different in particle numbers. On the other hand, we
show that it is possible to cut a countable set of Fock-state
components from the initial state with almost perfect fidelity,
thus effectively "combing" the input state and producing almost
pure state as the result.

The outline of the paper is as follows. In
Section~\ref{section0} we describe nonlinear coherent states  and
introduce a nonlinear coherent loss (NCL). Here also the the master
equation for coherent losses is given.

In Section~\ref{sec1} we investigate general properties of the NCL
and analyze their connection with features of a stationary state
which is to be obtained with the help of the considered kind of the
NCL.

In Section~\ref{sec2} we discuss the generation of pure Fock states and
consider the problem of coherence preserving for two important
classes of states that can be generated by the discussed approach
starting from a coherent state: superpositions of two Fock states
(either including vacuum state or not) and superpositions of states
with photon numbers distributed with equal intervals (e.g. states
with even or odd photon numbers).

In Section~\ref{sec3} the conditions, necessary for complete
preservation of coherence, are derived. It is shown that coherent
states taken as starting states for the considered type of
evolution do not satisfy the conditions exactly and, therefore,
cannot be used for creation of completely pure nonclassical states
by nonlinear absorption. However, an optimal \emph{classical}
starting state belongs to the class of coherent states, and the
final state, arbitrarily close to pure state, can be generated by
choosing sufficiently high amplitude of the starting coherent
state.

\section{Nonlinear coherent loss}
\label{section0}

In our work we consider dynamics of a single bosonic mode,
described by the following standard master equation in the
Lindblad form:
\begin{equation}
\label{eqn1} \begin{aligned} \frac{ d \rho(t)} {dt} =    \gamma
\left(2\hat L\rho(t)\hat L^+ - \hat L^+ \hat L \rho(t) - \rho(t)
\hat L^+ \hat L\right),
\end{aligned}
\end{equation}
where $\gamma >0$ is the decay rate and $\hat L$ is the Lindblad
operator. We assume that our bosonic mode is described by the
creation and annihilation operators $\hat a^+$ and $\hat a$. We
shall consider the non-unitary dynamics generated by the following
general class of Lindblad operators
 \begin{equation}
 \hat L={\hat a}f(\hat n)
 \label{operator f coherent }
 \end{equation}
where $\hat n = \hat a^+ \hat a$, and $f(n)$ is a non-negative
function.

The Lindblad operator (\ref{operator f coherent }) can be
considered as the annihilation operator of so-called $f$-deformed
harmonic oscillator with the commutation relations \cite{manko}
\[[L,L^{\dagger}]=(\hat n+1)(f(\hat n+1))^2-{\hat n}(f(\hat n))^2.
\]
Eigenstates of the operator $L$ were termed "nonlinear coherent
states" \cite{vogel} (it is curios to note that a specific
subclass termed "Mittag-Lefler coherent states" does actually
arise in micromasers in presence of loss and incoherent pump
\cite{tanya} ). For that reason, we refer to the decay described by the Lindblad operator (\ref{operator f coherent }) as "nonlinear coherent loss"
(NCL) in further consideration. It is interesting that any pure state non-orthogonal to an arbitrary Fock states can be exactly represented as a nonlinear coherent state \cite{davidovich-2001}. If it is orthogonal to some
Fock states (for example, if this pure state is a finite
superposition of Fock states), then one can still build a
nonlinear coherent state closely approximating the state in
question \cite{davidovich-2001}.

Thus, NCL looks highly promising for non-classical state
generation. In Ref. \cite{davidovich-2001} is was shown how to
build a function $f$ leading to an approximate generation of an
arbitrary Fock state from the initial coherent one. In more recent
work \cite{hong}, it was shown that the single-particle Fock
state, $|1\rangle$, can be generated with arbitrarily high
fidelity from the initial coherent state by the NCL with the
Lindblad operator (\ref{operator f coherent }) with $f(\hat
n)=\hat n-1$.

It should be emphasized that NCL design is completely
realistic and could be realized in practice. For vibrational
states of ions in magnetic traps NCL is already realized
\cite{{zoller},{wineland}}. Recently it was shown how to produce
the NCL by simple adjustment of well parameters in the three-well
trap configuration for Bose-Einstein condensates \cite{valera}.
Very recently the new realistic way of realizing NCL in
multi-core optical fiber was suggested \cite{mogilevtsev opt
lett}. It was shown that for experimentally realistic values of
Kerr nonlinearity of chalcogenide fibers with subwavelength core
(which is $10^5$ times higher than Kerr nonlinearity of a
conventional fused silica optical fiber \cite{{egg},{romanova}}),
it is possible to realize NCL in such a scheme and achieve a
deterministic generation of a single-photon state
\cite{mogilevtsev opt lett}.

Note that the losses and imperfections of the scheme,
generally, do not spoil the desired form of NCL (however,
they can lead to the appearance of other losses, both linear and
nonlinear, spoiling the effect of NCL). The form of NCL is
defined by the nonlinearity present in the scheme. For
illustration in the Appendix A an example of the nonlinear loss
appearance (and NCL appearance, in particular) is given for
the system of nonlinear bosonic modes coupled to the strongly
dissipative mode.

In this work we do not intend to discuss practical
realization of NCL in more details. Our aim is more fundamental; we want to
discuss the very possibilities offered by NCL. Further in this
work we discuss general limitations on the states that can be
generated by NCL starting from a coherent state. We demonstrate
when it is possible to generate pure superpositions of Fock
states, and when losses lead to the decrease of the states purity.

\section{General properties of function $f(n)$ and corresponding
stationary states} \label{sec1}

In order to analyze dependence between general properties of the
function $f(n)$, describing nonlinear absorption, and stationary
states that can be obtained as a result of evolution, characterized
by the master equation Eq.~(\ref{eqn1}), it is convenient to
decompose the density matrix in terms of Fock states:
\begin{equation}
\label{eqn2} \rho(t) = \sum_{n,m} \rho_{nm}(t)
\left|{n}\mathrel{\left\rangle{\vphantom{n m}}\right\langle
\kern-\nulldelimiterspace}{m}\right|.
\end{equation}
Then, the master equation  Eq.~(\ref{eqn1}) leads to the following
system of equations for density matrix elements:
\begin{equation}
\label{eqn3}\begin{aligned} \frac{d \rho_{nm}(t)} {dt} =
 2 \gamma F(n+1)
F(m+1) \rho_{n+1,m+1}(t) -{} \\ {} -  \gamma \{F^2 (n) + F^2(m)\}
\rho_{nm}(t),
\end{aligned}
\end{equation}
where $F(n) = \sqrt{n} f(n) \ge 0$, $F(0) = 0$.

Note that, according to the system (\ref{eqn3}), different
diagonals of the density matrix evolve independently. This fact
can be made more apparent by introducing notation
\begin{equation}
\label{eqn4} \rho_{n,n+k}(t) =c_k \xi_k(n,t),
\end{equation}
where constants $c_k$ can be arbitrary and will be fixed later.
Quantity $\xi_k(n,t)$ satisfies the following equation:
\begin{equation}
\label{eqn5} \begin{aligned} \frac{d \xi_k(n,t)} {dt} = 2 \gamma
F(n+1) F(n+k+1) \xi_k(n+1,t) -{} \\ {} -  \gamma \{F^2 (n) +
F^2(n+k)\} \xi_k(n,t),
\end{aligned}
\end{equation}
which does not contain $\xi_{k'}(n',t)$ for $k' \ne k$.

The density matrix is  the Hermitian one. It is completely defined
by elements $\rho_{nm}$ with $m\ge n$. Further, we will take into
consideration only the main diagonal of the density matrix and the
diagonals, lying \emph{below} the main diagonal. We assume $k\ge
0$ in Eqs.~(\ref{eqn4}), (\ref{eqn5}) and similar equations.

\newtheorem{lemma}{Lemma}
\begin{lemma}
\label{lemma1} If quantities $\xi_k(n,t)$ are positive
(non-negative) for all $n$ at $t=0$, they remain positive
(non-negative) for all $t>0$.
\end{lemma}
\begin{proof}
Eq.~(\ref{eqn5}) can be rewritten as
\begin{equation}
\label{eqn6}\begin{aligned} \frac{d \tilde\xi_k(n,t) }{dt}  = 2
\gamma F(n+1)F(n+k+1)\tilde\xi_k(n+1,t)   \times \\
{} \times e^{\gamma t (F^2(n) + F^2(n+k)- F^2(n+1) - F^2(n+k+1))}
,
\end{aligned}
\end{equation}
where $\tilde\xi_k(n,t) = \xi_k(n,t) e^{\gamma t (F^2(n) +
F^2(n+k))}$. Right-hand side of Eq.~(\ref{eqn6}) is non-negative
for positive or non-negative values of quantities
$\tilde\xi_k(n,t)$. It leads to non-negativity of derivatives on the
left-hand side of Eq.~(\ref{eqn6}) and non-decreasing character of
evolution of the quantities $\tilde\xi_k(n,t)$, if they are
initially non-negative. The exponential factor, connecting
$\tilde\xi_k(n,t)$ and $\xi_k(n,t)$ is strictly positive, and,
therefore, quantities $\xi_k(n,t)$ (together with
$\tilde\xi_k(n,t)$) preserve their positivity (non-negativity)
during evolution.
\end{proof}

If the initial state of the considered field mode is the coherent
state $|\alpha\rangle$, the condition of Lemma~\ref{lemma1} can be
satisfied by setting $c_k = (\alpha ^\ast)^k$ in Eq.~(\ref{eqn4}).
Then \[\xi_k(n,0) = \rho_{n,n+k} (0) / c_k = |\alpha|^{2n}
e^{-|\alpha|^2} / \sqrt{n! (n+k)!} > 0\] for all $n$, $k$.

It should be noted that any density matrix can be represented in
diagonal form in terms of coherent states using the
Glauber-Sudarshan $P$-representation \cite{glauber-1963}
\begin{equation}
\label{eqn10} \rho(0) = \int P(\alpha)
\left|{\alpha}\mathrel{\left\rangle{\vphantom{\alpha
\alpha}}\right\langle \kern-\nulldelimiterspace}{\alpha}\right|,
\end{equation}
where $P(x)$ is the Glauber-Sudarshan quasiprobability
distribution. Moreover, one can closely approximate the state in
question using a discrete set of coherent-state projectors on a
square lattice \cite{mogilevtsev prl}. Therefore, one may
conclude: all results derived here for an initial coherent state
conditions will be valid for an arbitrary initial state.
\begin{lemma}
\label{lemma2} If $\xi_k(n,0)>0$ for all $n$, for the stationary
value $\xi_k(n_1) = \lim_{t\rightarrow \infty} \xi_k(n_1,t)$ of
the quantity $\xi_k(n_1,t)$ to be non-zero, it is necessary that
$F(n_1) = F(n_1 + k)$.
\end{lemma}
\begin{proof}
Summing Eq.~(\ref{eqn5}) over $n$ leads to the following relation:
\begin{equation}
\label{eqn7} \frac{d}{dt} \sum_{n=0}^\infty \xi_k(n,t) = - \gamma
\sum_{n=0}^\infty \left\{ F(n) - F(n+k) \right\}^2 \xi_k(n,t),
\end{equation}
where we have taken into account that $F(0) = 0$. For the
stationary state the left-hand side of Eq.~(\ref{eqn7}) equals
zero. According to Lemma~\ref{lemma1}, each term on the right-hand
side of Eq.~(\ref{eqn7}) is non-negative, and the sum can be equal
to zero in stationary state, when for each $n$ either $\xi_k(n) =
\lim_{t\rightarrow \infty} \xi_k(n,t) = 0$ or $F(n) = F(n + k)$.
\end{proof}

\newtheorem{theorem}{Theorem}
\begin{theorem}
\label{theorem} Density matrix element $\rho_{n,n+k} (t)$ can have
non-zero value $ \lim_{t\rightarrow \infty} \rho_{n,n+k} (t) \ne
0$ in the stationary state of evolution, described by the master
equation Eq.~(\ref{eqn3}), only if
\begin{equation}
\label{eqn11} F(n) = F(n+k) = 0.
\end{equation}
\end{theorem}
\begin{proof}
If the initial state of the field mode is a coherent state $\rho(0)
=  \left|{\alpha}\mathrel{\left\rangle{\vphantom{\alpha
\alpha}}\right\langle \kern-\nulldelimiterspace}{\alpha}\right|$,
conditions of Lemmas~\ref{lemma1}, \ref{lemma2} are satisfied by
setting $c_k = (\alpha ^\ast)^k$. Therefore, for $
\lim_{t\rightarrow \infty} \rho_{n,n+k} (t) \ne 0$ it is necessary
that $F(n) = F(n+k)$. Eq.~(\ref{eqn5}) implies that the following
relation is satisfied:
\begin{equation}
\label{eqn8} \begin{aligned} \frac{d}{dt} \sum_{m=0}^{n-1}
\xi_k(m,t) = 2 \gamma F(n) F(n+k) \xi_k(n,t) -{} \\ {} - \gamma
\sum_{m=0}^{n-1} \left\{ F(m) - F(m+k) \right\}^2 \xi_k(m,t).
\end{aligned}
\end{equation}
In the limit $t\rightarrow \infty$ Eq.~(\ref{eqn8}) is transformed
into
\begin{equation}
\label{eqn9} 2 \gamma F^2(n) \xi_k(n) = 0,
\end{equation}
because, according to Lemma~\ref{lemma2}, for all $m$ either $F(m)
- F(m+k) = 0$ or $\xi_k(m) \equiv \lim_{t\rightarrow \infty}
\xi_k(m,t) = 0$. Therefore, for $ \lim_{t\rightarrow \infty}
\rho_{n,n+k} (t) = c_k \xi_k(n) \ne 0$, it is necessary that $F(n)
= F(n + k)$ and $F(n)=0$.

Taking into account that any initial state of the considered mode
can be represented using $P$-representation, one concludes that
 the quantity $\lim_{t\rightarrow \infty}
\rho_{n,n+k} (t)$ can have non-zero value, only if it is non-zero
in the stationary state for at least one coherent state $|\alpha
\rangle$ taken as the initial state. Then the first paragraph of
the Proof is applicable, and we obtain $F(n)=F(n+k) =0 $ as
necessary condition.
\end{proof}

In order to derive explicit expressions for  stationary values of
non-zero elements of the density matrix, one can rearrange system
of equations (\ref{eqn5}) in the following way. By introducing
quantities
\begin{equation}
\label{eqn12} T_k(n) = \frac {2 F(n) F(n+k)} {F^2(n)+ F^2(n+k)}
\end{equation}
Eq.~(\ref{eqn5}) can be transformed as
\begin{equation}
\label{eqn13} \begin{aligned} \frac{d \xi_k(n,t)} {dt} = \gamma
T_k(n+1) \Phi_k(n+1) \xi_k(n+1,t) -{} \\ {} -  \gamma \Phi_k(n)
\xi_k(n,t),
\end{aligned}
\end{equation}
where $\Phi_k(n) = F^2(n)+ F^2(n+k)$. Then, it is quite easy to
show that the following equality holds:
\begin{equation}
\label{eqn14}\begin{aligned} \frac{d}{dt} \sum_{m=n_1}^{n_2-1}
\xi_k(m,t) T_k(n_1+1) \cdot ... \cdot T_k(m-1) T_k(m) ={} \\{}  =  2
\gamma F(n_2) F(n_2+k) \xi_k(n_2,t) - \gamma\Phi_k(n_1)
\xi_k(n_1,t).
\end{aligned}
\end{equation}

Let the number $n_1$ correspond to the quantity $\xi_k(n_1,t)$,
which has non-zero stationary value $\xi_k(n_1) =
\lim_{t\rightarrow \infty} \xi_k(n_1,t)$ (i.e. the conditions
$F(n_1) = 0$ and $F(n_1 + k) = 0$ are satisfied). We can choose
$n_2$ to be the minimal number greater than $n_1$, for which at
least one of the following equations is satisfied: $F(n_2) = 0$ or
$F(n_2 + k) = 0$ ($n_2$ can be equal to infinity; then
$\xi_k(n_2,t) \rightarrow 0$). For this choice of $n_1$ and $n_2$
the right-hand side of Eq.~(\ref{eqn14}) equals zero. The
left-hand side of Eq.~(\ref{eqn14}) remains constant during
evolution. In the limit $t\rightarrow \infty$ we obtain
\begin{equation}
\label{eqn15}\begin{aligned} \xi_k(n_1) = \xi_k(n_1,0) +
\xi_k(n_1+1,0) T_k(n_1+1) + {} \\ {} + ... + \xi_k(n_2-1,0)
T_k(n_1+1) \cdot ... \cdot T_k(n_2-1).
\end{aligned}
\end{equation}
Using Eq.~(\ref{eqn4}) and returning to the density matrix
elements, one can rewrite Eq.~(\ref{eqn15}) as
\begin{equation}
\label{eqn16} \begin{gathered} \rho_{n_1,n_1+k}(\infty){} =
\rho_{n_1,n_1+k}(0) + {}
\\ {} + \rho_{n_1+1,n_1+k+1}(0) T_k(n_1+1) + ... + {}
\\  {} +  \rho_{n_2-1,n_2+k-1}(0) T_k(n_1+1) \cdot ... \cdot T_k(n_2-1).
\end{gathered}
\end{equation}

The obtained expression for non-zero elements of the stationary
density matrix can be interpreted in the following simple way. The
master equation (\ref{eqn1}) in the form (\ref{eqn3}) describes
"flow" of amplitudes of density matrix elements along diagonals in
the direction of photon number decreasing. The "transmittance" of
the transition between $\rho_{n,n+k}$ and $\rho_{n-1,n+k-1}$
equals $T_k(n)$ (see Eq.~(\ref{eqn13})). Elements $\rho_{n,n+k}$
with $F(n)=F(n+k) =0 $ "accumulate" the flow (i.e. they do not
transmit it to the next elements $\rho_{n-1,n+k-1}$). Elements
$\rho_{n,n+k}$ with either $F(n) =0 $ or $F(n+k) =0 $ (but without
the two conditions being satisfied simultaneously) neither
transmit the flow, nor accumulate it; for such elements $T_k(n) =
0$. This interpretation is illustrated in Fig.~\ref{fig1}.

\begin{figure}
\begin{center}
\begin{tabular}{cc}
\textbf{(a)} & \\
\multicolumn{2}{c}{\includegraphics[scale=0.5]{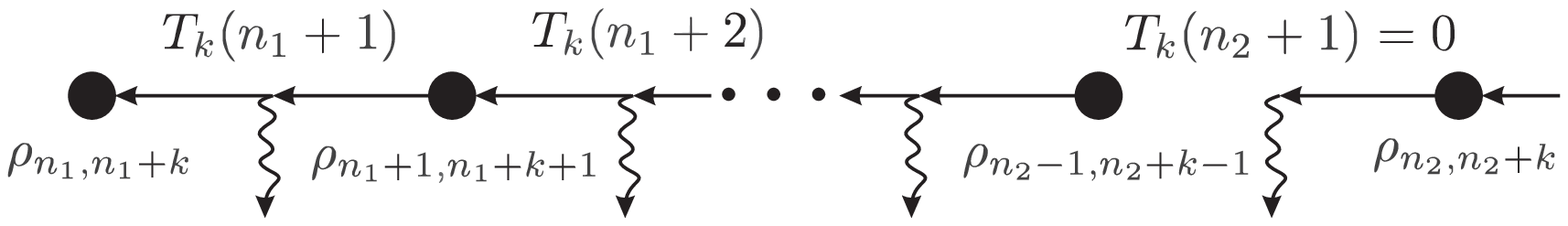}}\\ \\
\textbf{(b)} & \\
\multicolumn{2}{c}{\includegraphics[scale=1]{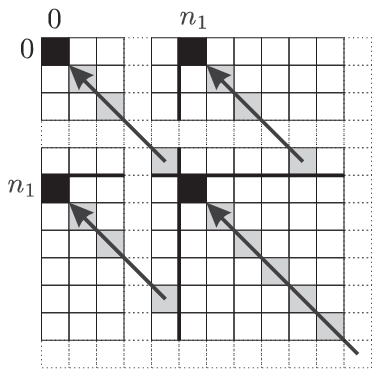}}\\ \\
\end{tabular}
\caption{\textbf{(a)} Interpretation of system evolution in terms
of "flow" of amplitude of density matrix elements. In the shown
example $F(n_1) = F(n_1 + k)=0$ and "flow" is accumulated by the
density matrix element $\rho_{n_1,n_1+k}$; $F(n_2) = 0$ and
$F(n_2+k) \ne 0$ or $F(n_2) \ne 0$ and $F(n_2+k) = 0$ and "flow"
is "absorbed" at the transition between $\rho_{n_2,n_2+k} $ and
$\rho_{n_2-1,n_2+k-1}$ without accumulation. \textbf{(b)}
Schematic representation of the density matrix for the case, when
$F(0) = F(n_1) =0$ and $F(n) \ne 0$ for $n\ne 0, n_1$. Black
squares represent density matrix elements, which are non-zero in
the stationary state. Grey squares represent elements, which give
contribute to non-zero elements of the stationary density matrix.
Arrows indicate direction of "flow". Thick lines correspond to
transitions with zero transmittances.}\label{fig1}
\end{center}
\end{figure}

It should be noted, that, according to the definition
Eq.~(\ref{eqn12}) of "transmittances" $T_k(n)$, amplitudes of
diagonal elements of density matrix are transmitted perfectly:
$T_0(n) \equiv 1$ (i.e. the trace is expectedly preserved).  In
order to attain maximal coherence, transmittances for non-diagonal
elements must be also equal to unity (at least those present in
expressions of the form of Eq.~(\ref{eqn16}) for non-zero elements
of the density matrix). Therefore, the condition of preserving
maximal coherence in the stationary state is
\begin{equation}
\label{eqn17} F(n)=F(n+k),
\end{equation} for density matrix elements giving non-zero
contributions to non-zero elements of the stationary density
matrix.

\section{Examples}
\label{sec2}

\subsection{Generation of Fock states}

The most "natural" nonclassical states that can be generated by
NCL with arbitrarily high fidelity are pure Fock states. Here and
further in this Section we assume that  the initial state is the
 coherent state, $|\alpha \rangle$, with the amplitude $\alpha$.

Suppose that the function $f(n)$ has only one zero $n_1$:
$f(n_1)=0$, $f(n) \ne 0$ for $n\ne n_1$. Then,  the function
$F(n)$ has two zeroes: $F(0)=0$ and $F(n_1) = 0$. According to the
Theorem~\ref{theorem}, only elements $\rho_{00}$, $\rho_{nn}$,
$\rho_{0n} = \rho_{n0}^\ast$ can have non-zero values in
stationary state for the system with such NLC. Eq.~(\ref{eqn16})
implies that  stationary values of on-zero diagonal elements of
the density matrix are described by following expressions:
\begin{equation}
\label{eqn57} \rho_{00}=\sum_{k=0}^{n_1-1} q_k^2(|\alpha|),
\end{equation}
where
\begin{equation}
q_m(|\alpha|) = |\alpha|^m e^{-|\alpha|^2 /2} / \sqrt{m!},
\label{q}
\end{equation}
and $\rho_{n_1n_1}=1- \rho_{00}$.

For large enough amplitudes, $|\alpha|$, of the initial coherent
state the following estimation is valid:
\begin{equation}
\label{eqn59} \rho_{00} < n_1 q_{n_1-1}^2
\xrightarrow[|\alpha|\rightarrow \infty]{}0.
\end{equation}
Therefore, the fidelity $F = \langle{n_1}\mathrel{|{\rho}|
\kern-\nulldelimiterspace}{n_1}\rangle = \rho_{n_1n_1}$ of
generating the Fock state $|n_1\rangle$ can be made arbitrarily
close to unity by choosing large enough amplitude $|\alpha|$ of
the starting coherent state (see Fig.~\ref{fig12}).

\begin{figure}
\begin{center}
\includegraphics[scale=0.7]{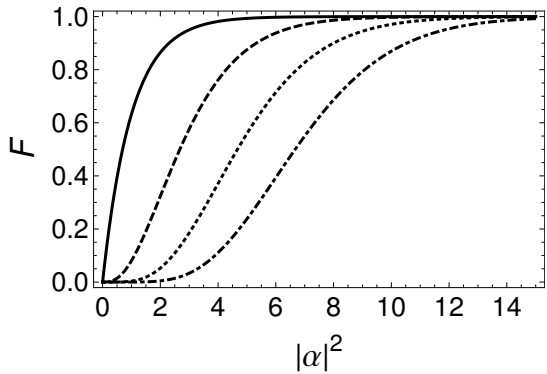}
\caption{Fidelity of generation of the Fock state $|n_1 \rangle$
by NCL starting from the coherent state $|\alpha\rangle$: solid,
dashed, dotted, dash-dotted lines correspond to $n_1=1,3,5,7$.}
\label{fig12}
\end{center}
\end{figure}

\subsection{Superposition of Fock states $|0\rangle$ and $|n\rangle$}

As the first example of coherent superpositions that can be
generated by NCL, we  consider  states maximally close to the
state
\begin{equation}
\label{eqn18} |\Psi_{0n} \rangle = \frac{ |0\rangle + e^{i \phi}
|n\rangle}{\sqrt 2}.
\end{equation}

If we require $F(n) = F(0) = 0$ and $F(m) \ne 0$ for $m\ne 0,n$,
only elements $\rho_{00}$, $\rho_{nn}$, $\rho_{0n} =
\rho_{n0}^\ast$ (see Eq.~(\ref{eqn2})) of the stationary density
matrix, $\rho$, will be non-zero (Theorem~\ref{theorem}). Then,
fidelity of generating the state Eq.~(\ref{eqn18}) is
\begin{equation}
\label{eqn19}\begin{gathered} F = \langle \Psi_{0n} \mathrel{|}
\rho \mathrel{|} \Psi_{0n} \rangle = {} \\  =  \frac{1}{2} +
\operatorname{Re}\left(\rho_{0n}e^{i \phi} \right) \le \frac{1}{2}
+\left|\rho_{0n}\right|.
\end{gathered}
\end{equation}
 The maximal possible value
of fidelity is $F =\frac{1}{2} +\left|\rho_{0n}\right|$.
Positivity of density matrix implies that $|\rho_{0n}| \le
\frac{1}{2}$. For convenience, we shall characterize coherence of
the stationary density matrix by the quantity
\begin{equation}
\label{eqn21} c_{0n} = 2 |\rho_{0n}|, \quad c_{0n} \in [0,1].
\end{equation}
So, maximization of the coherence $c_{0n}$ leads also to
maximality of the fidelity $F$.

For the initial coherent state amplitude $\alpha = |\alpha| e^{i
\phi /n}$ the following equality holds: $\rho_{0n}(0) =
|\rho_{0n}(0)| e^{-i \phi}$. According to Lemma~\ref{lemma1}, this
holds also for any moment of time, $t$. Therefore, in order to
derive conditions for preserving maximal fidelity by the nonlinear
absorption, one needs to carry out maximization of the quantity
$c_{0n}$ over amplitudes $|\alpha|$.

According to Eqs.~(\ref{eqn16}), (\ref{eqn21}) (see also
Fig.~\ref{fig1}(b)), the coherence $c_{0n}$ equals to
\begin{equation}
\label{eqn22} c_{0n}= 2 \sum_{m=0}^{n-1} q_m(|\alpha|)
q_{m+n}(|\alpha|) T_n(m),
\end{equation}
where $q_m(|\alpha|)$ are given by Eq. (\ref{q}). To maximize the
coherence $c_{0n}$ one needs to have $T_n(m) = 1$ for
$m=0,...,n-1$. According to Eq.~(\ref{eqn17}), this condition
implies that
\begin{equation}
\label{eqn23} F(m) = F(m+n)\quad \mbox{for} \quad m=0,...,n-1.
\end{equation}

Further maximization of Eq.~(\ref{eqn22}) can be carried out
numerically. Results of numerical calculations are shown in
Fig.~\ref{fig2}. One can see that the best performance of the
method is achieved for $n=2$. Then, according to
Eq.~(\ref{eqn22}), the coherence is
\begin{equation}
\label{eqn24} c_{02} = \sqrt{2} \left( |\alpha|^2 + |\alpha|^4 /
\sqrt{3} \right) e^{-|\alpha|^2 }.
\end{equation}
The modulus of optimal amplitude of the initial coherent state is
$|\alpha_{opt}| = \sqrt{\frac{1}{2}
\left(2-\sqrt{3}+\sqrt{7}\right)} \approx 1.2$. Elements of the
final density matrix in the optimal case are $\rho_{00} = 0.60$,
$\rho_{22} = 0.40$, $|\rho_{02}| = 0.44$ (see Fig.~\ref{fig2}(b),
and the coherence is $c_{02} = 0.88$.

\begin{figure*}
\begin{center}
\begin{tabular}{cc}
\textbf{(a)} & \textbf{(b)} \\ \\
\includegraphics[scale=0.7]{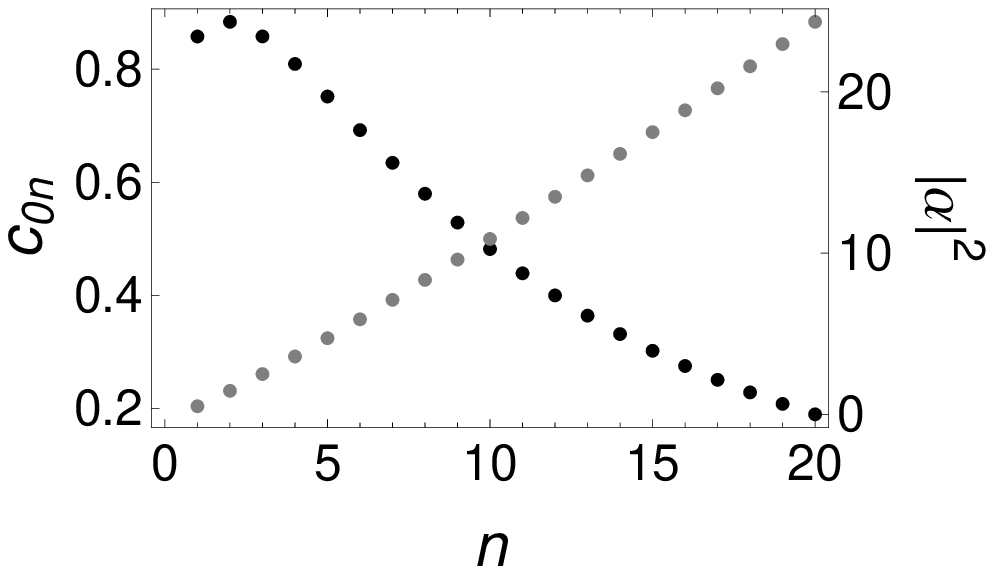} &
\includegraphics[scale=0.6]{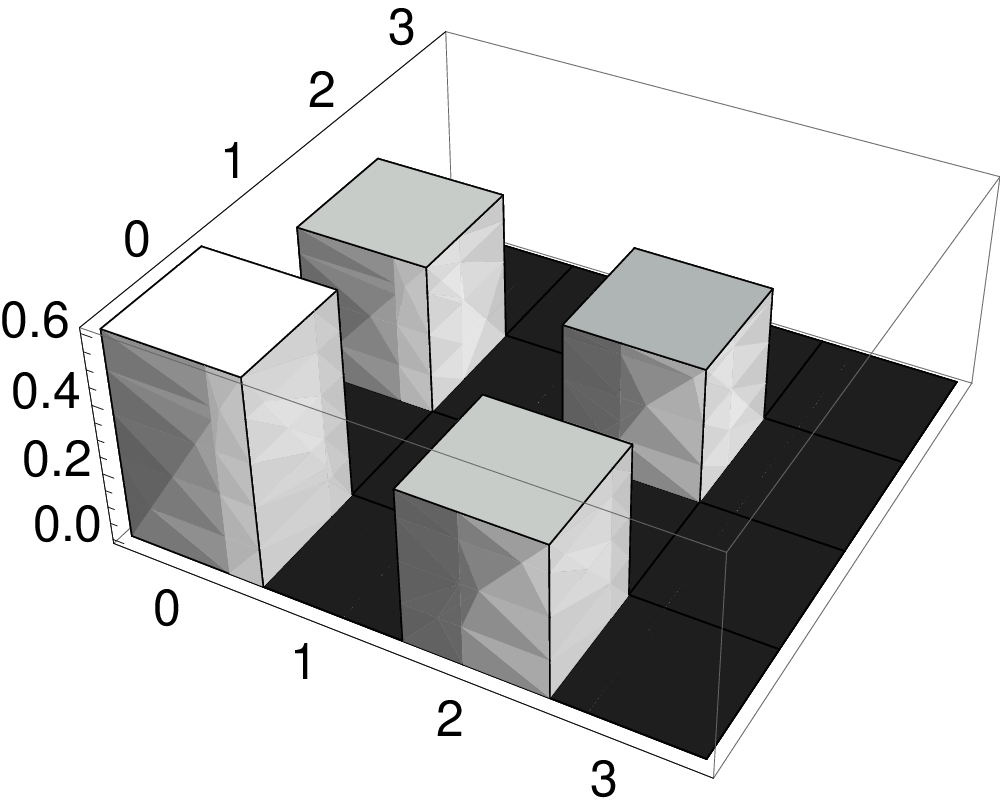}\\ \\
\end{tabular}
\caption{\textbf{(a)} The modulus of the optimal amplitude of the
starting coherent state (grey dots) and the maximal achievable
coherence (black dots) for generation of the state $|\Psi_{0n}
\rangle$. \textbf{(b)} Absolute values of the stationary state
density matrix generated for the initial coherent state optimal
for generation of $|\Psi_{02} \rangle$ .}\label{fig2}
\end{center}
\end{figure*}

\subsection{Superpositions of Fock states $|m\rangle$ and $|n\rangle$}

The next example is the generation of a state maximally close to
the state
\begin{equation}
\label{eqn25} |\Psi_{nm} \rangle = \frac{ |n\rangle + e^{i \phi}
|m\rangle}{\sqrt 2}.
\end{equation}

Notice that there is rather pronounced difference between the
current example and the one considered in the previous Subsection.
According to definition of the function $F(n)$,  $F(0) = 0$ holds
for any system. For any amplitude $\alpha$ of the initial coherent
state the element $\rho_{00}$ of the initial density matrix has
non-zero value. Therefore, according to Eq.~(\ref{eqn16}), this
density matrix element will have non-zero value in the final
state. However, its value for the ideal state $|\Psi_{nm} \rangle$
must be equal to zero.

In order to generate the state sufficiently close to $|\Psi_{nm}
\rangle$, we require $F(m)= F(n) = 0$, $F(k) \ne 0$ for $k\ne
n,m,0$. Only the final density matrix elements $\rho_{ij}$ with
$i,j = 0,n,m$ are non-zero in the stationary state. Similarly to
the previous example, fidelity of the desired state is
\begin{equation}
\label{eqn26} \begin{gathered} F = \langle \Psi_{nm} \mathrel{|}
\rho \mathrel{|} \Psi_{nm} \rangle = {} \\ {} =   \frac{1}{2} +
\operatorname{Re}\left(\rho_{nm}e^{i \phi} \right) - \rho_{00} \le \frac{1}{2}
+\left|\rho_{nm}\right|  - \rho_{00}.
\end{gathered}
\end{equation}
Again, we define the coherence
\begin{equation}
\label{eqn27}  c_{nm} = 2 |\rho_{nm}|, \quad c_{nm} \in [0,1].
\end{equation}
Further, we use coherence, rather than fidelity, for characterization of the state preparation quality.

Phase of the optimal initial coherent state in this case is
defined by the following equation:
\begin{equation}
\label{eqn28} \alpha = |\alpha| e^{i \phi /(m-n)}.
\end{equation}
For preserving maximal coherence and for obtaining
maximal fidelity of the generated state, we need
\begin{equation}
\label{eqn29} F(k) = F(k+m-n)\quad \mbox{for} \quad k=0,...,m-n-1
\end{equation}
(see Eqs.~(\ref{eqn17}), (\ref{eqn23})).

If above conditions are satisfied, the coherence $c_{nm}$ is
described by the following expression:
\begin{equation}
\label{eqn30} c_{nm} = 2 \sum_{k=n}^{m-1} q_k(|\alpha|)
q_{k+m-n}(|\alpha|).
\end{equation}

Optimal values of $|\alpha|$ and maximal achievable coherence
$c_{nm}$ can be found either numerically (dots in
Fig.~\ref{fig3}(a,b)), or by approximate analytical expressions
(lines in Fig.~\ref{fig3}(a,b)).

\begin{figure*}
\begin{center}
\begin{tabular}{cc}
\textbf{(a)} & \textbf{(b)} \\ \\
\includegraphics[scale=0.7]{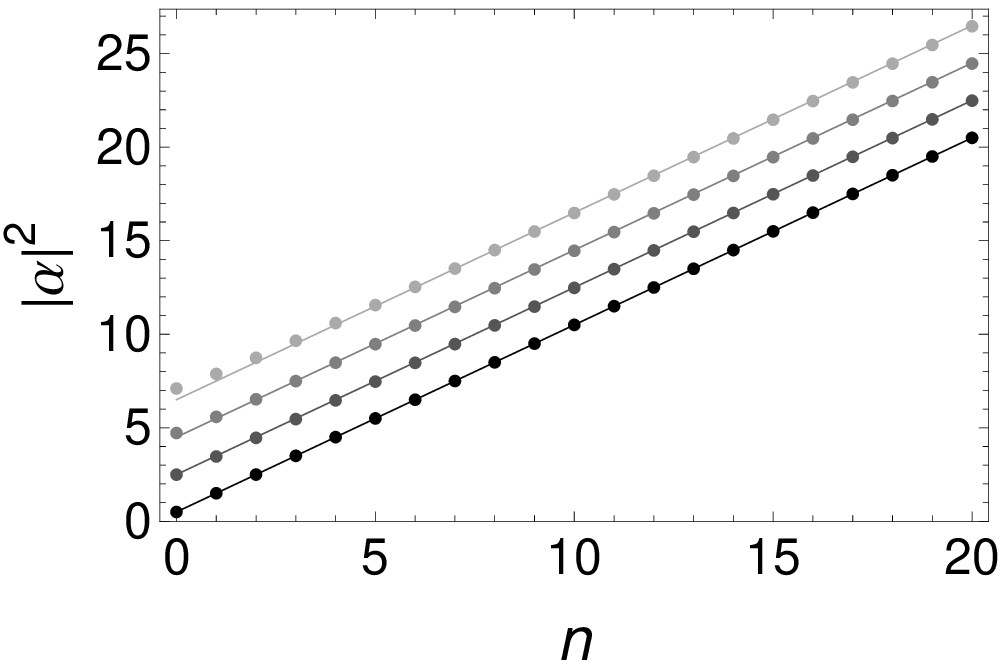} &
\includegraphics[scale=0.7]{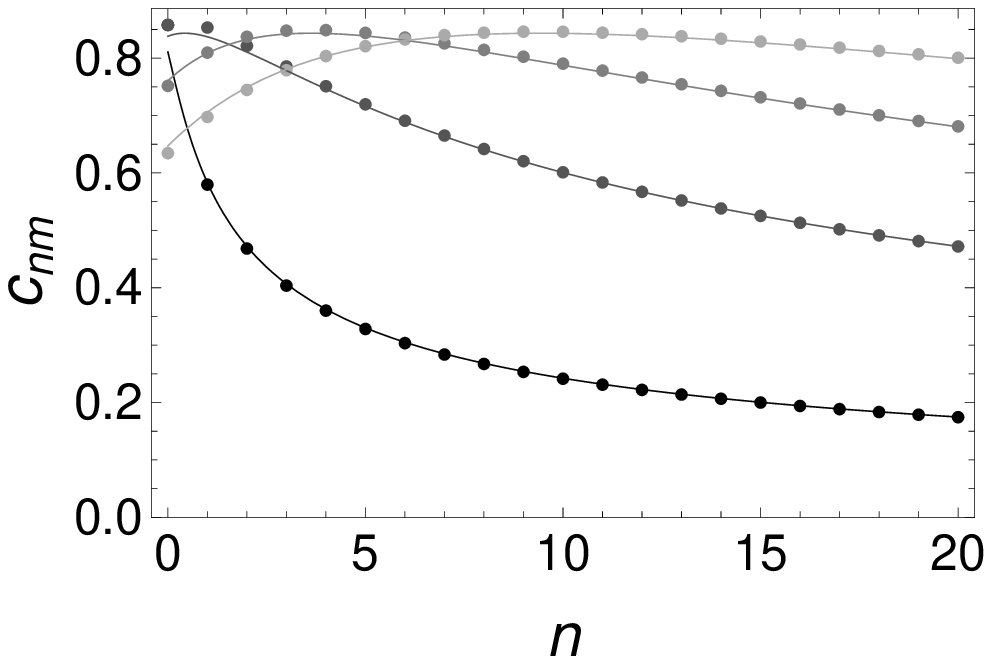}\\ \\
\textbf{(c)} & \textbf{(d)} \\ \\
\includegraphics[scale=0.6]{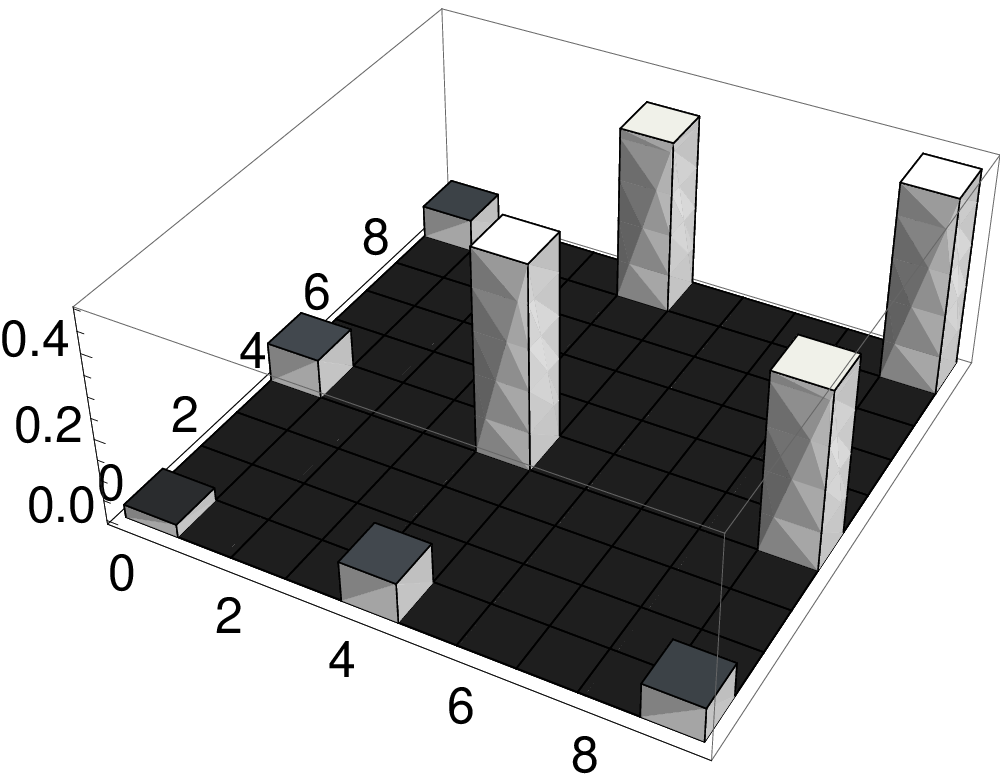} &
\includegraphics[scale=0.6]{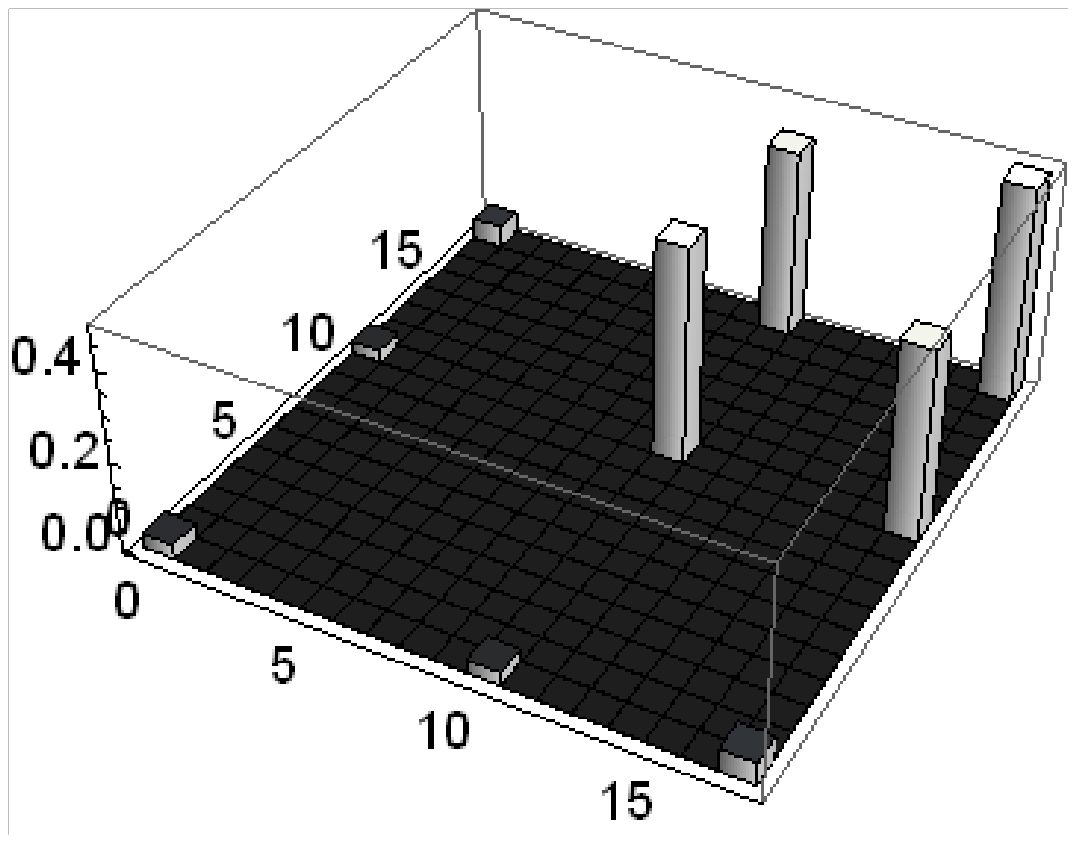}\\ \\
\end{tabular}
\caption{\textbf{(a), (b)} Optimal amplitude of the starting
coherent state (plot (a)) and maximal achievable coherence (plot
(b)) for generation of the state $|\Psi_{nm} \rangle$: dots and
lines corresponds to numerical optimization and approximate
analytical expressions (Eq.~(\ref{eqn33}), Eq.~(\ref{eqn34}));
black, dark grey, grey, light grey  dots and lines correspond to
$m=n+1,n+3,n+5,n+7$. \textbf{(c), (d)} Absolute values of elements
of stationary state density matrices generated by the schemes
optimized for creation of the state $|\Psi_{4,9} \rangle$ (plot
(c)) and $|\Psi_{10,17} \rangle$ (plot (d)).}\label{fig3}
\end{center}
\end{figure*}

To obtain an analytical expression for the coherence $c_{nm}$, we
will take into account that (according to central limit theorem),
for large enough values of $|\alpha|$ the following approximation
is valid for $q_k (|\alpha|)$ given by Eq.(\ref{q}):
\begin{equation}
\label{eqn31} q_k^2 (|\alpha|) \approx \frac{1}{\sqrt{2 \pi}
|\alpha|} \exp \left\{ - \frac{ (k-|\alpha|^2)^2}{2 |\alpha|^2}
\right\}.
\end{equation}
Then, the coherence is approximately equal to
\begin{equation}
\label{eqn32} \begin{aligned} &c_{nm} = \frac{2}{\sqrt{2 \pi}
|\alpha|} \times\\ &\times \sum_{k=n}^{m-1} \exp \left\{ - \frac{
(k-|\alpha|^2)^2 + (k+m-n-|\alpha|^2)^2}{4 |\alpha|^2} \right\}.
\end{aligned}
\end{equation}
To maximize the coherence one needs practically the same
$|\alpha|$, as required for maximalization of the numerator in the
exponent of Eq.~(\ref{eqn32}). This value is
\begin{equation}
\label{eqn33} |\alpha_{opt}|^2 \approx m- \frac{1}{2}.
\end{equation}

Taking into account Eq.~(\ref{eqn33}) and approximating summation
by integration in Eq.~(\ref{eqn32}), one can derive the following
expression for the coherence:
\begin{equation}
\label{eqn34} c_{nm} \approx 2 \operatorname{erf} \left(
\frac{\Delta n} {2 \sqrt{2} |\alpha_{opt}|} \right) \exp \left( -
\frac{\Delta n^2} {8 |\alpha_{opt}|^2} \right),
\end{equation}
where $\Delta n = m-n$. Fig.~\ref{fig3} shows that obtained
approximate expressions represent results of numerical
calculations with sufficiently high accuracy.

The expression Eq.~(\ref{eqn34}) for the coherence $c_{nm}$ has
maximal value 0.84 for $\Delta n^2 \approx 6.4 |\alpha_{opt}|^2$.
Therefore, the method is most suitable for generation of the
states with
\begin{equation}
\label{eqn35} (m-n)^2 \approx 6.4 (m-{\textstyle\frac{1}{2}}).
\end{equation}

Using the approximation Eq.~(\ref{eqn31}), one can show that the
density matrix elements of the optimal stationary state are
described by the following expressions:
\begin{equation}
\label{eqn36} \begin{gathered} \rho_{mm} =\textstyle
\frac{1}{2},\quad \rho_{nn} = \frac{1}{2} \operatorname{erf} ( 2
\zeta), \\
|\rho_{nm}| = \operatorname{erf} (\zeta) e^{-\zeta^2},\quad
\rho_{00} = \textstyle \frac{1}{2}\left\{ 1-  \operatorname{erf} ( 2
\zeta) \right\},
\end{gathered}
\end{equation}
where $\zeta = \frac{\Delta n} {2 \sqrt{2} |\alpha_{opt}|}$.

Examples of density matrices of the states generated by the
schemes optimized for $n=4$, $m=9$ and $n=10$, $m=17$, are shown
in Fig.~\ref{fig3}(c),(d).

\subsection{Superpositions of states with equidistant photon numbers}

Here we consider examples of functions $f(n)$ (and $F(n)$) with
countable sets of zeros. We focus on the case of  functions with
equidistantly distributed zeros:
\begin{equation}
\label{eqn37} f(jN+n_0) = 0, \quad j=0,1,2,...,
\end{equation}
where $N$ is the distance between neighboring zeros and $n_0$, $0
\le n_0 < N$, determines the position of the first zero. It should
be noted that, according to the definition of $F(n)$, one also has
$F(0)=0$ for any function $f(n)$.

For such kind of NCL only  elements $\rho_{nm}$ with $n,m = 0,
n_0, n_0+N, n_0+2N, ...$ remain non-zero in the stationary state.

As in  previously discussed examples, for preserving the maximal
coherence we require that the function $F(n)$ should satisfy
Eq.~(\ref{eqn17}). This implies
\begin{equation}
\label{eqn38} F(n+N) = F(n),\mbox{ for } n=1,2,...
\end{equation}
i.e. function $F(n)$ must be periodic.

In previous examples zeros of $F(n)$ were giving a clue for
choosing the amplitude for the initial coherent state. It is not
so in the current case. Now  the amplitude $|\alpha|$ will be
considered as a free parameter, and the final state will be
investigated as a function of the amplitude. For the sake of simplicity,
we  assume that the amplitude $\alpha$ is real and positive.

According to Eq.(\ref{eqn16}), non-zero elements of the stationary
density matrix are described by the following expression:
\begin{equation}
\label{eqn39} \rho_{nm} = \sum _{k=0}^{N-1} q_{n+k}(|\alpha|)
q_{m+k}(|\alpha|),
\end{equation}
\begin{equation}
\label{eqn40} \rho_{0n} = \sum _{k=0}^{n_0} q_{k}(|\alpha|)
q_{n+k}(|\alpha|), \mbox{ if } n_0 \ne 0.
\end{equation}

For large amplitudes of the initial coherent state $|\alpha|^2 \gg
N$ the approximate expression for $q_n(|\alpha|)$, provided by
Eq.~(\ref{eqn31}), can be used to simplify Eq.~(\ref{eqn39}). One
can show that in this case
\begin{equation}
\label{eqn41} \rho_{nm} = N q_{n}\left(|\alpha'| \right)
q_{m}\left(|\alpha'| \right) \cdot \left\{ 1 + O\left(
\frac{1}{|\alpha|^2} \right)\right\},
\end{equation}
where $|\alpha'|^2 = |\alpha|^2 - (N-1)/2$. Therefore, the
stationary density matrix can be approximated with
\begin{equation}
\label{eqn42} \rho \approx |{\Phi_N^{(n_0)}}\mathrel{\rangle\langle
\kern-\nulldelimiterspace}{\Phi_N^{(n_0)}}|,
\end{equation}
where
\begin{equation}
\label{eqn43} |\Phi_N^{(n_0)}\rangle = const \cdot
\sum_{j=0}^\infty \frac{{\alpha'}^{jN+n_0}}{\sqrt{(jN+n_0)!}}
e^{-|\alpha'|^2/2} |jN+n_0\rangle
\end{equation}
is the state that can be obtained from a coherent state
$|\alpha'\rangle$ by retaining only the states with photon numbers
$n_0$, $n_0+N$, $n_0 + 2N$, ... On the other hand, the state
$|\Phi_N^{(n_0)}\rangle$ can be considered as a superposition of
coherent states, distributed on a circle:
\begin{equation}
\label{eqn44} |\Phi_N^{(n_0)}\rangle = const \cdot \sum_{k=0}^{N-1}
e^{-2 \pi i k n_0 /N} |\alpha' e^{2 \pi i k/N} \rangle.
\end{equation}
For example, for $N=2$ the states $|\Phi_2^{(0)}\rangle$ and
$|\Phi_2^{(1)}\rangle$ are superpositions of states with even and
odd numbers of photon respectively and correspond to the following
superpositions of coherent states with opposite phases:
$|\Phi_2^{(0)}\rangle \sim |\alpha'\rangle+ |-\alpha'\rangle$ and
$|\Phi_2^{(1)}\rangle\sim |\alpha'\rangle- |-\alpha'\rangle$.

To describe "quality" of generation of the superpositions, it is
convenient to consider purity of the final state, defined as
\begin{equation}
\label{eqn45} P = \operatorname{Tr} \rho^2 = \sum_{nm} \rho_{nm}
\rho_{mn}.
\end{equation}
According to Eq.~(\ref{eqn39}), for large $|\alpha|$, when presence
of non-zero element $\rho_{00}$ can be neglected, this expression
can be rewritten as
\begin{equation}
\label{eqn46} P= \sum_{k=0}^{N-1} \left(\sum_n q_{n+k}^2 \right)^2 +
\sum_{k_1 \ne k_2} \left( \sum_n q_{n+k_1} q_{n+k_2} \right)^2,
\end{equation}
where $n=n_0,n_0+N,n_0+2N,...$. Approximating summation by
integration and using Eq.(\ref{eqn31}), one can find the following
expression for the state purity, valid for $|\alpha|^2 \gg N$:
\begin{equation}
\label{eqn47} P \approx 1 - \frac{N^2 - 1}{24 |\alpha|^2}.
\end{equation}
Therefore, in the limit $|\alpha| \rightarrow \infty$ the obtained
stationary state tends to a pure state (Fig.~\ref{fig4}(a)).

\begin{figure*}
\begin{center}
\begin{tabular}{cc}
\textbf{(a)} & \textbf{(b)} \\ \\
\includegraphics[scale=0.7]{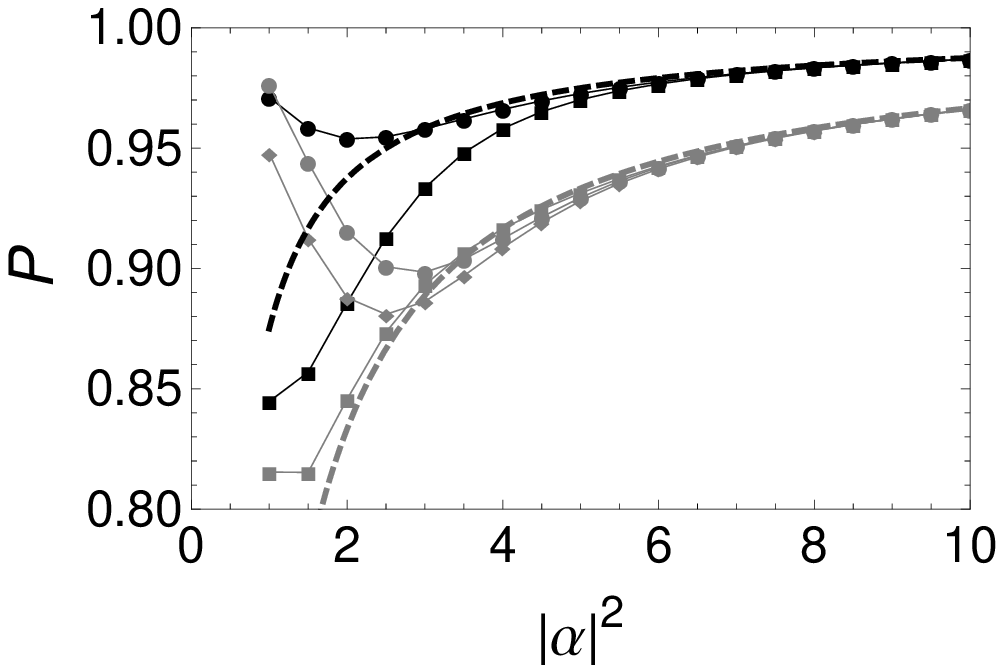} &
\includegraphics[scale=0.6]{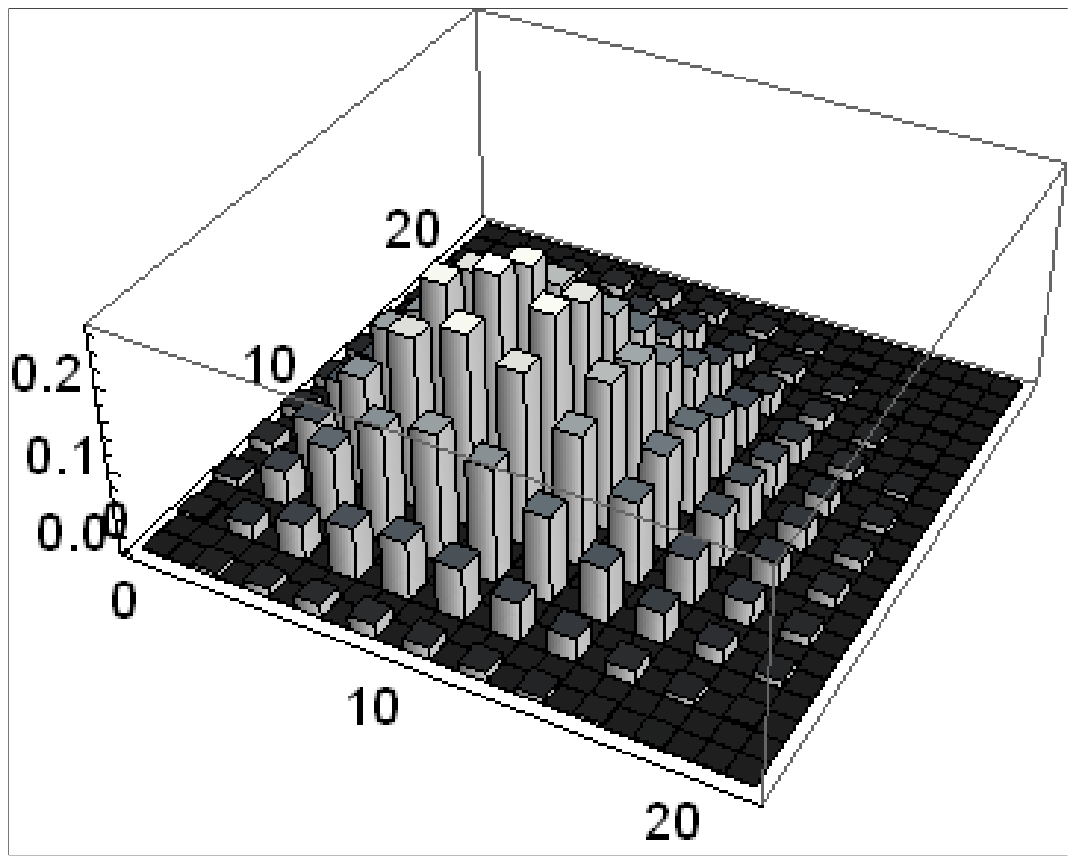}\\ \\
\textbf{(c)} & \textbf{(d)} \\ \\
\includegraphics[scale=0.6]{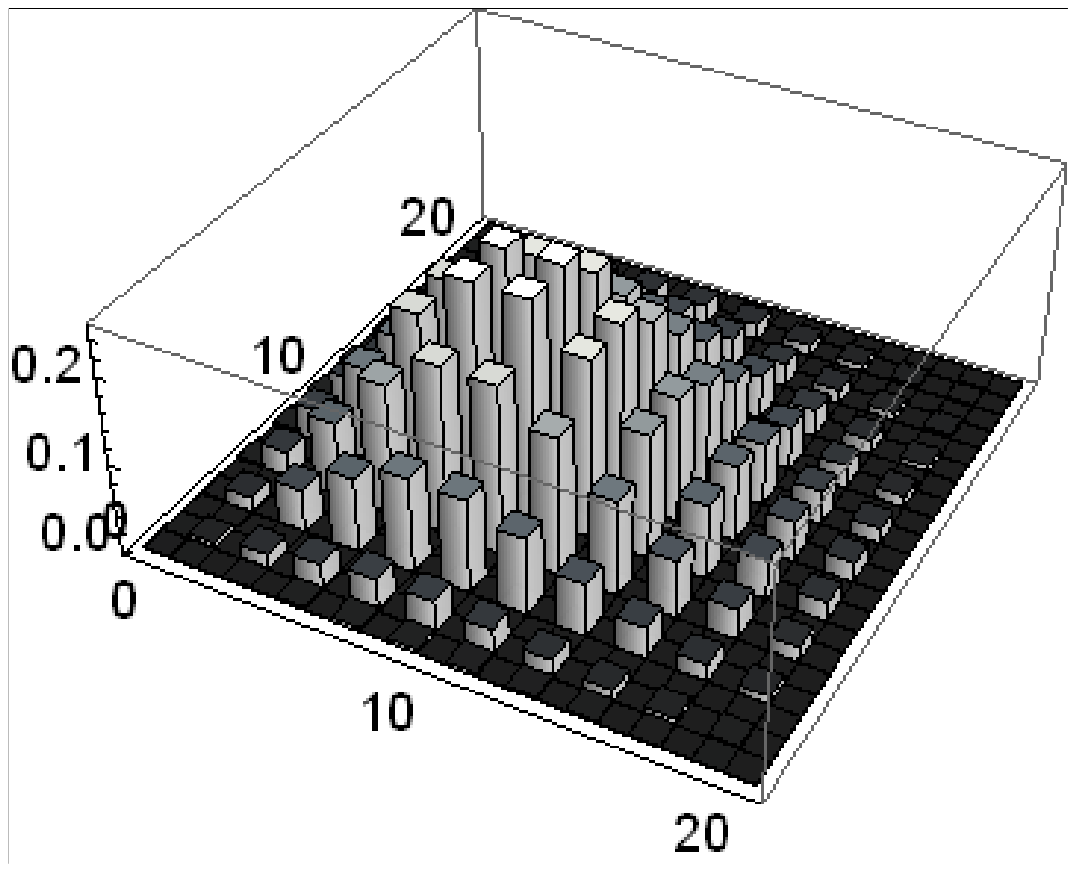} &
\includegraphics[scale=0.6]{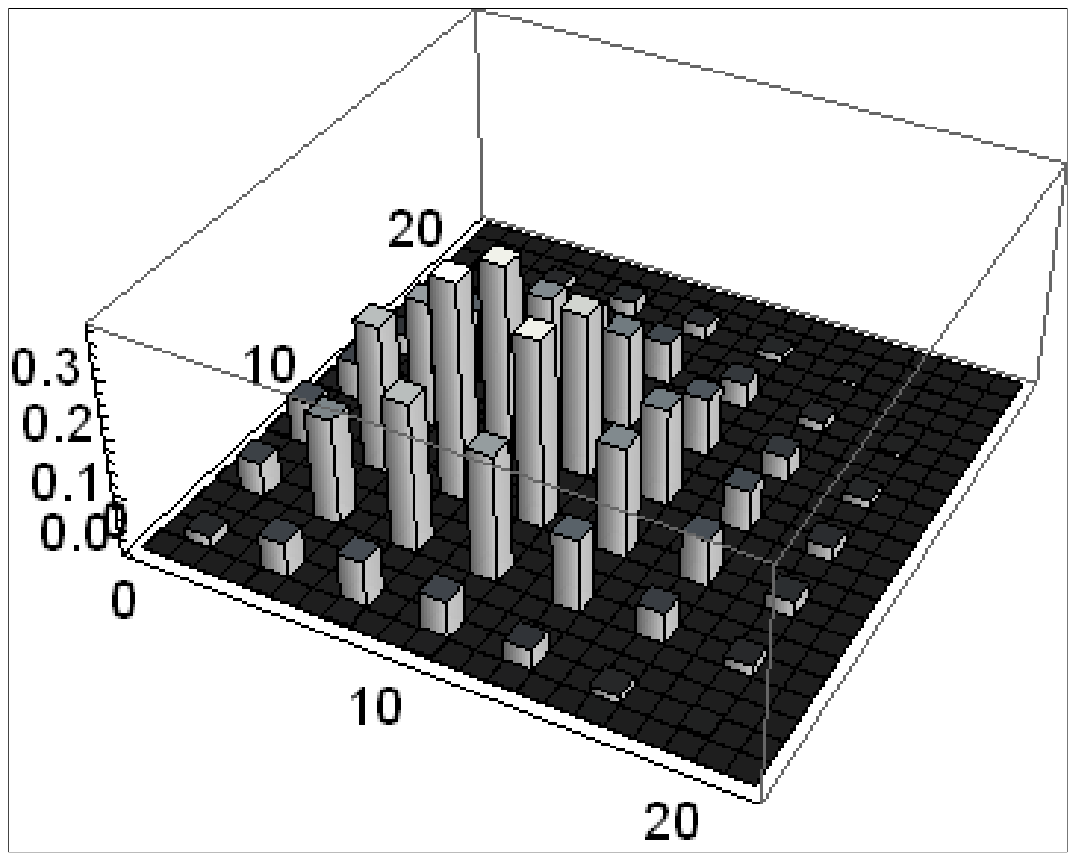}\\ \\
\end{tabular}
\caption{\textbf{(a)} Purity of stationary states, generated by
schemes with the function $F(n)$ having equidistant zeroes: dots
(connected by solid line) correspond to numerical optimization,
dashed lines correspond to the approximate analytical expression
(Eq.~(\ref{eqn47})); black dots and lines, grey dots and lines are
for $N=2,3$; circles, squares, rhombuses are for $n_0 = 0,1,2$.
\textbf{(b)} Matrix elements of the stationary state with even
photon numbers, generated from coherent state with $|\alpha|^2 =
9$ ($N=2$, $n_0=0$). \textbf{(c)} Matrix elements of the
stationary state with odd photon numbers, generated from coherent
state with $|\alpha|^2 = 9$ ($N=2$, $n_0=1$). \textbf{(d)} Matrix
elements of the stationary state with photon numbers, distributed
with interval 3 ($|\alpha|^2 = 9$, $N=3$, $n_0=1$).} \label{fig4}
\end{center}
\end{figure*}

Several examples of density matrix that can be generated by the
method are shown in Fig.~\ref{fig4}(b--d).

\section{Discussion}
\label{sec3}

The discussed examples show that systems with appropriate
nonlinear losses can be used for creation of Fock states
superpositions with sufficiently high fidelity. However, for all
examples obtained stationary states were mixed. It is quite
interesting to find the conditions that must be satisfied for the
final stationary state to be pure and to determine whether such
conditions can be fulfilled.

As stated above, Eq.~(\ref{eqn17}) is one of the conditions
necessary for complete coherence preserving.

Several other conditions are provided by  following lemmas.

\begin{lemma}
\label{lemma3} If all diagonal elements of the initial state
density matrix are positive, the final stationary state can be
pure only if the function $F(n)$ has equidistantly distributed
zeroes.
\end{lemma}
\begin{proof}
Suppose $n_1$, $n_1'$ and $n_2$, $n_2'$ are two pairs of successive
zeroes of the function $F(n)$, and $n_1'-n_1 < n_2' - n_2$. Either
of $n_1'$, $n_2'$ can be equal to infinity, if the corresponding
zero $n_1$ or $n_2$ is the last zero of $F(n)$.

Then, according to Eq.~(\ref{eqn16}), the density matrix elements
have the following values in the stationary state:
\begin{equation}
\label{eqn48} \rho_{n_1,n_1} = \sum_{k=0}^{n_1'-n_1-1}
\rho_{n_1+k,n_1+k}(0),
\end{equation}
\begin{equation}
\label{eqn49} \rho_{n_1,n_2} = \sum_{k=0}^{n_1'-n_1-1}
\rho_{n_1+k,n_2+k}(0),
\end{equation}
\begin{equation}
\label{eqn50} \begin{aligned} \rho_{n_2,n_2} =
\sum_{k=0}^{n_2'-n_2-1} \rho_{n_2+k,n_2+k}(0) >{} \\  {} >
\sum_{k=0}^{n_1'-n_1-1} \rho_{n_2+k,n_2+k}(0) ,
\end{aligned}
\end{equation}
where we have taken into account positivity of diagonal elements
of the initial state density matrix. On the other hand, the
following inequality holds:
\begin{equation}
\label{eqn51}\begin{gathered} \Bigl( \sum_{k} \rho_{n_1+k,n_1+k}(0)
\Bigr) \Bigl(\sum_{k} \rho_{n_2+k,n_2+k}(0))\Bigr) \ge \\ \ge
\Bigl(\sum_{k} \sqrt{ \rho_{n_1+k,n_1+k}(0) \rho_{n_2+k,n_2+k}(0)}\Bigr)^2 \ge \\
\ge \Bigl| \sum_k \rho_{n_1+k,n_2+k}(0) \Bigr|^2,
\end{gathered}
\end{equation}
when all the sums are taken over the same range of $k$.

Therefore, assumption of nonequal distances $n_1'-n_1$ and
$n_2'-n_2$ between successive zeroes $n_1$, $n_1'$ and $n_2$, $n_2'$
leads to the following inequality:
\begin{equation}
\label{eqn52} \rho_{n_1,n_1} \rho_{n_2,n_2} > | \rho_{n_1,n_2}|^2,
\end{equation}
which manifests that the state cannot be pure. Finally, the
distance between any neighboring zeroes must be equal to some
constant value, which was denoted by $N$ in the last example of
Section~\ref{sec2}.
\end{proof}

\begin{lemma}
Let $n_1$, $n_1'$ and $n_2$, $n_2'$ be two pairs of successive
zeroes of the function $F(n)$, $n_1'-n_1 = n_2' - n_2$, and let
all diagonal elements of the initial state density matrix be
positive. The final stationary state can be pure only if the
density matrix of the initial state satisfies the following
condition:
\begin{equation}
\label{eqn53}\frac{ \rho_{n_1+k,n_1+k}(0)}{ \rho_{n_2+k,n_2+k}(0)} =
\frac{\rho_{n_1,n_1}(0)}{ \rho_{n_2,n_2}(0)} \mbox{ for }
k=0,...,n_1'-n_1-1.
\end{equation}
\end{lemma}
\begin{proof}
Elements $\rho_{n_1,n_1}$, $\rho_{n_1,n_2}$, $\rho_{n_2,n_2}$ of
the stationary state density matrix are described by
Eqs.~(\ref{eqn48}),(\ref{eqn49}) and the first line of
Eq.~(\ref{eqn50}).

The condition
\begin{equation}
\label{eqn54} \rho_{n_1,n_1} \rho_{n_2,n_2} = | \rho_{n_1,n_2}|^2,
\end{equation}
which is necessary for the state purity, is satisfied only if all
parts of Eq.~(\ref{eqn51}) are equal to each other. Obviously,
equality can be achieved only if "vectors" $\{
\rho_{n_1+k,n_1+k}(0) \}$ and $\{ \rho_{n_2+k,n_2+k}(0) \}$ are
collinear, which implies Eq.~(\ref{eqn53}).
\end{proof}

On the basis of these lemmas the following theorem can be proved.

\begin{theorem}
If the function $F(n)$ has at least one zero $n_1$, $n_1>0$, and
the initial state is classical, the final stationary state will be
mixed.
\end{theorem}
\begin{proof}

Any classical initial state can be represented as a mixture of
coherent states with positive weights. Therefore, if we prove that
the stationary state will be mixed for an arbitrary initial
coherent state, the stationary state will be proven to be mixed
for any classical starting state.

For the initial state being a coherent state $|\alpha \rangle$ the
following statement holds:
\begin{equation}
\label{eqn55} \frac{ \rho_{n_1+k,n_1+k}(0)}{ \rho_{n_2+k,n_2+k}(0)}
= |\alpha|^{2(n_2-n_1)} \frac{(n_2+k)!}{(n_1+k)!}.
\end{equation}
Therefore, Eq.~(\ref{eqn53}) is not satisfied except for the trivial
case, when $n_1'-n_1 = 1$, $F(n) \equiv 0$.
\end{proof}

However, when the amplitude $|\alpha|$ of the initial coherent
state is large enough, the generated state can be very close to a
pure state. Indeed, only the density matrix elements
$\rho_{n_1,n_2}(0)$ with $|n_1-|\alpha|^2| \lesssim |\alpha|$,
$|n_2-|\alpha|^2| \lesssim |\alpha|$ have significantly nonzero
values. For $n_1,n_2 \gg k$ Eq.~(\ref{eqn55}) can be transformed
into
\begin{equation}
\label{eqn56}\begin{gathered} \frac{ \rho_{n_1+k,n_1+k}(0)}{
\rho_{n_2+k,n_2+k}(0)} \approx |\alpha|^{2(n_2-n_1)}
\left(\frac{n_2}{n_1}\right)^{k} \approx \\ \approx
|\alpha|^{2(n_2-n_1)} \left(1+ k \frac{n_2-n_1}{n_1} \right) = \\ =
|\alpha|^{2(n_2-n_1)} \left\{1+ O\left(\frac{1}{|\alpha|}\right)
\right\}.
\end{gathered}
\end{equation}
Therefore, Eq.~(\ref{eqn53}) can be satisfied with arbitrarily
high precision by using high enough amplitudes of the starting
coherent state.

\section{Conclusions}

We have investigated the nonlinear coherent loss as a resource for
generating non-classical states. We have established conditions
for generating a prescribed state (namely, arbitrarily Fock states
and specific superpositions of them) from the initial coherent
state. We have highlighted a connection between properties of the
Lindblad operator of NCL and the generated state. We have
demonstrated that the state generated by NCL from the initial
classical state will always be mixed. However, for certain classes
of states one can generate almost pure states. In some cases it is
possible to reach high fidelity of the generation by appropriately
choosing an amplitude of the initial coherent state. Fock states
belong to such a class; one can generate an almost pure Fock state
for an initial coherent state of a sufficiently high amplitude.
The situation complicates for finite superpositions of Fock
states. For example, a superposition of two Fock states cannot be
generated by NCL with the arbitrarily high fidelity from the
initial coherent state; a strict upper border exists for this
case. However, certain infinite superpositions can also be
generated with an arbitrarily high fidelity. For example, one can
"comb" the initial coherent state cutting off Fock state
components with odd number of particles with an arbitrarily high
fidelity.

Finally, it should be noted that though coherent states are
"nonideal" for generating states via the NCL, they remain optimal
initial classical states. Any classical state is a mixture of
coherent states with positive weights. Quantum mechanical
equations for evolution of the density matrix  are always linear.
Therefore, any final state obtained from a classical state, is
necessarily a mixture of final states which are to be obtained
from corresponding initial coherent states. Thus, no classical
initial state can lead to purity of the final state, greater than
the purity, provided by "the best choice" from the possible the initial coherent states.

The authors acknowledge the financial support by the BRFFI of
Belarus. They are very grateful to V. S. Shchesnovich for helpful
discussions.

\begin{appendix}
\section{Designing the NCL by correlated loss}

To highlight the concept of designing nonlinear loss in systems of
coupled bosonic modes (which can be realized, for example, in
Bose-Einstein condensates \cite{valera} or in optical fibers
\cite{mogilevtsev opt lett}), let us consider a model of $N+1$
bosonic modes, $a_1\ldots a_{N+1}$ coupled in the usual linear way
to the same Markovian reservoir. Thus, the model is described by
the master equation in the standard Lindblad form:
\begin{equation}
\frac{d}{dt}\rho_{N+1}=-\frac{i}{\hbar}[H_{N+1},\rho_{N+1}]+
\gamma\mathcal{L}_{N+1}\rho_{N+1}, \label{lindblad multimode}
\end{equation} where $H_{N+1}$ is the Hamiltonian describing
unitary evolution of modes described by the annihilation, $a_j$, and
creation, $a_j^{\dagger}$, operators; $\gamma >0$ is the linear
decay rate. Here and in the derivation below we define
superoperators $\mathcal{L}_j$ on the basis of Lindblad operators
$L_j$ in the following way:
\[\mathcal{L}_j\rho=2L_j\rho L_j^{\dagger} - L_j^{\dagger}L_j\rho
-\rho L_j^{\dagger}L_j.
\]
In Eq.~(\ref{lindblad multimode}) the Lindblad operator is
\[L_{N+1} \equiv b_{N+1}=\sum\limits_{j=1}^{N+1}u_{N+1,j}a_j,
\quad \sum\limits_{j=1}^{N+1}|u_{N+1,j}|^2=1. \] It depends linearly
on the annihilation operators $a_j$ and, therefore, represents a
collective mode. One can introduce a new set of independent bosonic
operators, $b_j$, by the unitary transformation $u_{i,j}$:
\[b_{k}=\sum\limits_{j=1}^{N+1}u_{k,j}a_j,\]
where coefficients $u_{N+1,j}$ are fixed by the interaction between
the system and the reservoir. The Hamiltonian, $H_{N+1}$, can be
decomposed in terms of annihilation, $b_{N+1}$, and creation,
$b_{N+1}^\dagger$, operators of the collective mode:
\begin{equation}
H_{N+1}=\sum\limits_{m,n=0}^{\infty}
F_{mn}\left(b_{N+1}^{\dagger}\right)^mb_{N+1}^n,
\label{hamiltonian transformed}
\end{equation}
where operators $F_{m,n}$ depend only on operators $b_j$ and
$b_j^\dagger$ for $j<N+1$.

Now we assume that the state of the collective mode $b_{N+1}$ decays
to the vacuum very rapidly on the time-scale of dynamics prescribed
by the Hamiltonian $H_{N+1}$.Thus, an adiabatic elimination of the
mode $b_{N+1}$ can be made resulting in the following equation for
the reduced density matrix of modes $b_j$ for $j=1\ldots N$:
\begin{equation}
 \frac{d}{dt}\rho_{N}=-\frac{i}{\hbar}[F_{0,0},\rho_{N}]+
\sum\limits_{n=1}^{\infty} \frac{n!}{(n+1)\gamma}
\mathcal{L}_{n}\rho_{N}, \label{lindblad multimode reduced}
\end{equation}
where Lindblad operators are $L_n=F_{0,n}$. The
master equation (\ref{lindblad multimode reduced}) describes both
possible interaction between modes $b_j$ and their nonlinear
losses. Note that $F_{0,0}$ might include nonlinearities not
present in the original Hamiltonian, $H_{N+1}$.

Even for low order nonlinearities the scheme described above is
able to lead to the appearance of NCL. In Ref.
\cite{mogilevtsev opt lett} it was shown that the scheme
(\ref{lindblad multimode}) for the case of just two bosonic modes
subject to Kerr nonlinearity leads to appearance of NCL with
$\hat L={\hat a}\hat n$.

\end{appendix}

\end{document}